\documentclass[12pt, twoside, reqno, openany]{amsart}

\usepackage[top=1.5cm, bottom=1.5cm, left=2.0cm, right=2.0cm]{geometry}
\usepackage[english]{babel}
\usepackage[T1]{fontenc}
\usepackage[utf8]{inputenc} 
\usepackage[colorlinks, citecolor=black, linkcolor=black]{hyperref}
\usepackage{amsfonts}
\usepackage{amssymb}
\usepackage{amsmath}
\usepackage{amsthm}
\usepackage{array}
\usepackage{bm}
\usepackage{braket}
\usepackage{mathtools}
\usepackage{esint}

\usepackage{enumerate}
\usepackage{setspace}
\usepackage{dsfont}
\usepackage[pdftex]{graphicx}
\usepackage{graphicx}
\usepackage{float}
\usepackage{color}
\usepackage[labelsep=none]{caption}
\usepackage{chngcntr}
\counterwithin{figure}{section}
\usepackage[utf8]{inputenc}
\usepackage[dvipsnames]{xcolor}
\usepackage{amsaddr}

\newtheorem{theorem}{Theorem}[section]

\theoremstyle{remark}
\newtheorem{remark}[theorem]{Remark}

\def\R {\mathbb{R}}

\title{Ground state energy of Bogoliubov energy functional in the high density limit}
\author[N. Mokrza\'nski]{Norbert Mokrza\'nski$^*$}
\address{Department of Mathematical Methods in Physics, Faculty of Physics,\\ University of Warsaw, Pasteura 5, 02-093 Warszawa, Poland}
\email{norbert.mokrzanski@fuw.edu.pl}
\thanks{$^\ast$Corresponding author.}

\author[B. Pałuba]{Bartosz Pałuba}
\address{Institute of Geophysics, Faculty of Physics, \\University of Warsaw, Pasteura 5, 02-093 Warszawa, Poland}
\email{bartosz.paluba@fuw.edu.pl}

\date{May 27, 2024}

\begin{document}
\DeclareRobustCommand{\uvec}[1]{{%
    \bm{\mathbf{e_{#1}}}%
}}

\DeclareRobustCommand{\vec}[1]{{%
    \bm{\mathbf{#1}}%
}}

\newcommand{\norm}[1]{\left|#1\right|}
\newcommand{\dd}{\text{d}}
\newcommand{\qRq}{\quad\Rightarrow\quad}
\renewcommand{\labelenumi}{\roman{enumi})}

\maketitle

\begin{abstract}
We consider the Bogoliubov energy functional proposed by Napiórkowski, Reuvers and Solovej and analize it in the high density regime. We derive a two term  asymptotic expansion of the ground state energy.
\end{abstract}


\section{Introduction}
In 1947 Nikolay Bogoliubov published the work "On the theory of superfluidity" \cite{Bogo}, which turned out to be a fundamental result in the theory of interacting Bose gases. The main goal of that  paper was a microscopic derivation of Landau's criterion of superfluidity, which is quantum effect  observed in liquid helium \textsuperscript{4}He. Bogoliubov, using a heuristic argument, proposed that the Hamiltonian of the system could be replaced by an effective, simpler one. There were two main assumptions made by Bogoliubov. The first one was that there is Bose-Einstein Condensation in the system. The second one was that the strength of interaction between two particles outside the condensate is negligible. Hence, to obtain an effective Hamiltonian, terms involving more than two creation or annihilation operators of non-zero modes can be dropped. This led to an effective quadratic Hamiltonian, which could be diagonalized explicitly, leading to an excitation spectrum that allowed to explain the superfluidity phenomenon using the Landau's argument \cite{Landau}.  

Bogoliubov's approach, although it was lacking mathematical rigor, appeared to be very useful for describing bosonic interacting systems. When Bose-Einstein condensates had begun to be obtained in experiments involving cold atoms \cite{Exp1, Exp2}, this topic became particularly interesting in both theoretical and mathematical physics communities. Since then significant progress regarding understanding of this theory has been made (see \cite{Napiorkowski-23,Schlein-22} for recent reviews).

In this paper we want to revisit the approach proposed by Napiórkowski, Reuvers and Solovej in \cite{NapReuSol1, NapReuSol2, NapReuSol3, FouNapReuSol-19} (and before that by Critchley and Solomon in \cite{CritSolo}). The authors analyze the thermodynamic properties of the interacting Bose gas system using variational perspective based on Bogoliubov theory. They introduce the free energy (density) functional $\mathcal{F}$ obtained by evaluating the free energy expectation value on quasi-free states and passing to the macroscopic limit (see next section for more detailed discussion). In the first paper \cite{NapReuSol1} the authors prove the existence of the minimizers under certain assumptions on the interaction potential $V$ (in \cite{Old} it has been shown that some of those assumptions can be weakened). They also prove that this model predicts a phase transition in the system. In the second paper \cite{NapReuSol2} authors consider the dilute limit and obtain results concerning the critical temperature and the expansion of the free energy in terms of temperature and total density $\rho$ of the system. In \cite{FouNapReuSol-19} the authors considered the dilute two-dimensional model and obtained a ground state energy expansion. The main reason why the authors analyzed this variational model is that the many-body problem (i.e. starting from the many-body Schr\"odinger equation) in the thermodynamics limit is very difficult and the results are rather limited (see  \cite{Haberbergeetall-23,Bastietall-22,FouSol-20} for most recent ones).

The goal of this paper is to investigate the opposite, high density limit of the Bogoliubov functional in the case of zero temperature. From the many-body perspective, this problem was for the first time considered by Lieb in \cite{Lieb}. In that paper he argued how the two leading terms of the ground state energy should look like. The lower bound is in fact a straightforward consequence of the Onsager lemma (see \cite[Lemma 7.2]{DerNap-13}). Lieb also claims the upper bound can be obtained using modified Girardeu states (see \cite{Girardeau}), however computational details are missing there. We plan to provide an alternative approach for this problem in a future work, in the spirit of \cite{NamNap}. Later, the high density asymptotics have been derived within an effective model introduced by Lieb and then analyzed by Carlen, Jauslin and Lieb \cite{CJL-19,CHJL-20,Ja-22}. Our analysis is motivated by the question of one of the authors above who asked whether the Bogoliubov functional predicts the high density free energy expansion in agreement with the works above. We show that this is, indeed, the case. Our second result shows that imposing some extra conditions on the regularity of the potential leads to obtaining better, more precise estimations of the error term.

The paper is structured as follows: in the next section we will sketch the derivation of the functional and formulate the main result. The section following it is devoted to the proof of the main theorem. In the last section we obtain additional error estimates.

\section{The functional and the main result}

We will start with a brief sketch of the derivation of the Bogoliubov energy functional. We refer to the appendix of \cite{NapReuSol1} for full details of the derivation, also for the non-zero temperature case concerning the entropy.

The Hamiltonian of the interacting Bose gas in the box of size $L$ (in the momentum representation) is given by:
\begin{equation}
    H = \sum_{p\in \frac{2\pi}{L}\mathbb{Z}^d} \left(p^2-\mu\right)a_p^*a_p + \frac{1}{2L^3}\sum_{p,q,k\in \frac{2\pi}{L}\mathbb{Z}^d}\widehat{V}(k)a_{p+k}^*a_{q-k}^*a_pa_{q},
    \label{Hamiltonian}
\end{equation}
where $p$ denotes momentum, $\mu$ is a chemical potential of the gas, $\widehat{V}(k)$ is a Fourier coefficient of the (periodized) potential function $V$ and $a_p$, $a_p^*$ are respectively annihilation and creation operators of a particle with momentum $p$. Inspired by the fact that the ground states of quadratic Hamiltonians are quasi-free states, we will restrict our attention to such states only. By definition, these are the states that satisfy Wick rule and therefore are fully determined by density function $\gamma$ and function of pairing $\alpha$:
\begin{equation*}
    \gamma(p)=\langle a_p^*a_{p}\rangle, \quad \alpha(p)=\langle a_pa_{-p}\rangle
\end{equation*}
Without loss of generality we can also assume that $\langle a_pa_{-p}\rangle = \langle a_p^*a_{-p}^*\rangle$, so that $\alpha(p)$ is real-valued.

In order to include the possibility of Bose-Einstein condensation in the system, we use Weyl transformation to introduce the density of the condensate $\rho_0$. This operation transforms the annihilation operators as:
$$a_p \to a_p + \delta_{0,p}\sqrt{L^3\rho_0},$$
analogously for the creation operators. 

Having the Weyl transformed Hamiltonian, we consider its expectation value in the quasi-free, translation invariant state. Dividing this expectation value by the volume of the system $L^3$ and formally passing to the macroscopic limit (that is taking $|\Lambda| \to \infty$ and assuming $\frac{1}{|\Lambda|}\sum_p \to \frac1{(2\pi)^3}\int_{\R^3}\dd p$) we conclude that the energy (density) may be described by the functional $\mathcal{F}$ given by the formula:
\begin{equation}
\label{functional}
    \begin{split}
    \mathcal{F}(\gamma,\alpha,\rho_0) &= \left(2\pi\right)^{-3} \int_{\mathbb{R}^3}p^2\gamma(p)\dd p  + \frac12\widehat{V}(0)\rho^2 \\&
    + \rho_0\left(2\pi\right)^{-3}\int_{\mathbb{R}^3}\widehat{V}(p)\left[\gamma(p)+\alpha(p)\right]\dd p
    \\&
    +\frac12 \left(2\pi\right)^{-6}\iint_{\mathbb{R}^3\times \mathbb{R}^3}\widehat{V}(p-q)\left[\alpha(p)\alpha(q)+\gamma(p)\gamma(q)\right]\dd p\dd q.
    \end{split}   
\end{equation}
Here $\rho$ is the total density of the system, defined as:
\begin{equation}
    \label{total_density}
    \rho = \rho_0 + \left(2\pi\right)^{-3}\int_{\mathbb{R}^3}\gamma(p)\dd p =: \rho_0 + \rho_\gamma,
\end{equation}
where $\rho_\gamma$ was introduced to denote the density of the particles outside the condensate. For our purpose, the convention of the Fourier transform and its inverse is:
\begin{equation*}
    \widehat{V}(p) = \int_{\mathbb{R}^3}V(x)e^{-ipx}\dd x, \quad V(x) = \left(2\pi\right)^{-3}\int_{\mathbb{R}^3}\widehat{V}(p)e^{ipx}\dd p.
\end{equation*}
The domain of the functional is as follows:
\begin{equation}
    \label{domain}
    \mathcal{D} = \left\{(\gamma, \alpha, \rho_0) \colon \ \gamma\in L^1\left(\left(1+p^2\right)\dd p\right), \ \gamma(p)\geq 0,\ \alpha(p)^2\leq\gamma(p)^2+\gamma(p),\ \rho_0\geq 0\right\}.
\end{equation}
The inequalities in the above definition correspond to the physical assumptions that the number of particles cannot be negative and the associated generalized two-particle reduced density matrix has to be positive definite.

The goal of this paper is to study asymptotic behavior of the canonical ground state energy, i.e. the quantity:
\begin{equation}
\label{GSE_def}
E(\rho) = \inf_{\mathcal{D}_\rho} \mathcal{F}(\gamma, \alpha, \rho_0), 
\end{equation}
where $\mathcal{D}_\rho$ is the subset of the domain with total density \eqref{total_density} fixed to be equal $\rho$. Note that this is not a standard meaning of a canonical formulation of the problem, as it is only the expectation value of the number of particles that is fixed. We are interested in deriving the leading order terms of the asymptotic expansion of the energy in the high density regime $\rho \to \infty$ and give an estimation of the error. Our result can be considered as the counterpart to the result of \cite{NapReuSol2}, where the authors investigated the behavior of $\mathcal{F}(\gamma,\ \alpha,\ \rho_0)$ in the dilute limit $\rho^{1/3}a \ll 1$, where $a$ is a scattering length of the interaction potential $V$.

The existence of both canonical and grand canonical (i.e. when the total density $\rho$ of the system is not fixed) minimizers for this functional was proven in \cite{NapReuSol1} under certain assumptions on the potential $V$. In \cite{Old} it was shown that those assumptions in the grand canonical version can be reduced to:
\begin{equation}
\label{V_assumption}
V \ge 0, \; V \in L^1(\R^3), \; V \not \equiv 0, \; \widehat V \ge 0, \; \widehat V \in L^1(\R^3).
\end{equation}
In this paper we will work under the same assumptions on the potential $V$. After the proof of the first theorem we will discuss the possibilities of weakening the assumptions even further (cf. Remark \ref{V_not_rep} at the end of next section). We would like to emphasize that our result is completely independent of the existence of the minimizers, yet we impose assumptions \eqref{V_assumption} as they happen to be well suited for the considered problem.

Before stating the main result, to avoid misinterpretation, we will recall the notation used to denote asymptotic errors: we say that some function $h(x)$ is of order $o(f(x))_{x \to \infty}$ if $\lim_{x \to \infty} \frac{f(x)}{h(x)} = 0$. Similarly $h(x)$ is of order $O(f(x))$ if $\frac{f(x)}{h(x)}$ stays bounded for $x \to \infty$.

Now we can state the theorems.

\begin{theorem}
\label{main_theorem}
Assume the potential $V$ satisfies \eqref{V_assumption}. The ground state energy $E(\rho)$ defined as in \eqref{GSE_def} in the limit of high densities $\rho \to \infty$ can be expanded as:
\begin{equation}
    \label{asymptotics}
    E(\rho) = \frac12\widehat{V}(0) \rho^2 - \frac12V(0)\rho + o(\rho)_{\rho \to \infty}.
\end{equation}
\end{theorem}
The leading terms in the above expansion correspond to the ones proposed in \cite{Lieb} (where the claim was phrased in terms of energy per particle). To prove this theorem we will establish suitable lower and upper bound of the ground state energy $E(\rho)$ -- this will be done in the next section. Before that, we will formulate the next result concerning the error order estimation.

\begin{theorem}
\label{extra_regularity}
Let $V$ be a function satisfying assumptions \eqref{V_assumption}.
\begin{enumerate}
\item Assume furthermore that $\widehat V$ decays polynomially at infinity, that is, for sufficiently large $|p|$, as:
\begin{equation}
\label{rate_of_conv}
\widehat{V}(p) \leq \frac{C}{|p|^{k}},
\end{equation}
for some constants $C > 0$, $k > 3$. Then there exist a $\theta \in \left(\frac23, 1\right)$ (dependent on $k$)\\*such that:
\begin{equation}
\label{asymptotics_regular}
E(\rho) = \frac12\widehat{V}(0) \rho^2 - \frac12V(0)\rho + O(\rho^\theta)_{\rho \to \infty} .
\end{equation}
The exact value of $\theta$ can be chosen as:
$$\theta = \frac{2k}{3k - 3}.$$

\item Assume that instead of \eqref{rate_of_conv} function $\widehat V$ decays exponentially, i.e. satisfies inequality
$$\widehat{V}(p) \leq Ce^{-c|p|}$$
for some constants $C>0$ and $c>0$. Then
\begin{equation}
\label{asymptotics_regular_exp}
E(\rho) = \frac12\widehat{V}(0) \rho^2 - \frac12V(0)\rho + O\left(\rho^{2/3} (\ln \rho)^6\right)_{\rho \to \infty} .
\end{equation}
\end{enumerate}

\end{theorem}

As we will see in the following sections, the proof of the Theorem \ref{extra_regularity} will follow by obtaining a stronger estimation on one of the terms in the functional. Let us proceed to the proof of the first theorem.

\section{Proof of Theorem \ref{main_theorem}}

As mentioned before, the proof is based on derivation of proper lower and upper bound of $E(\rho)$.

\subsection{The lower bound}

We will start with the derivation of the lower bound, as it is rather straightforward. Note that the only term possibly yielding negative values in the functional \eqref{functional} is:
\begin{equation}
\label{negative}
\rho_0 (2\pi)^{-3}\int_{\R^3} [\alpha(p) + \gamma(p)] \dd p.
\end{equation}
Indeed, every term aside the quadratic one in $\alpha$ is positive due to the assumptions $\hat V \ge 0$ and $\gamma \ge 0$. To prove positivity of this particular quadratic term we will show that
\begin{equation}
\label{alpha_conv_positive}
(2\pi)^{-6}\iint_{\mathbb{R}^3\times \mathbb{R}^3}\widehat{V}(p-q)\alpha(p)\alpha(q)\dd p\dd q = \int_{\R^3} V(x) |\check \alpha (x)|^2 \dd x,   
\end{equation}
which is positive as $V \ge 0$. As $\alpha$ is not necessarily a $L^2$ function, we cannot use Plancherel Theorem directly. We need to proceed more carefully. First note that $\alpha \in L^1(\R^3) + L^2(\R^3)$ due to the decomposition $\alpha = \alpha_1 + \alpha_2$ with $\alpha_1 = \alpha \chi_{\gamma > 1} \in L^1$ and $\alpha_2 = \alpha \chi_{\gamma \le 1} \in L^2$. It follows that the inverse Fourier transform $\check \alpha$ is well defined. Now we can rewrite $\alpha(p)\alpha(q)$ appearing inside the integral on the left hand side of \eqref{alpha_conv_positive} as
$$\alpha(p)\alpha(q)= \alpha_1(p)\alpha_1(q) + \alpha_1(p)\alpha_2(q) + \alpha_2(p)\alpha_1(q) + \alpha_2(p)\alpha_2(q).$$
With this expansion we can express the last three terms as $L^2$ inner products of $\overline{\alpha}_i$ and the convolutions $\widehat V * \alpha_j$ ($i,j=1,2$). This is possible as the convolutions $\widehat V * \alpha_2$ and $\widehat V * \alpha_1$ are $L^2$ functions due to the facts that $\widehat V \in L^1$, $\widehat V \in L^2$ (by interpolation with the $L^\infty$ norm) and Young inequality. Now, for those terms, Plancherel Theorem is applicable. As for the first term, we can express it by integral of $V$ and inverse transform of $\alpha$ by using the definition of Fourier transform for $V$ and changing the order of integration by Fubini theorem (we use the assumptions on $V$). In the end we obtain
\begin{align*}
&(2\pi)^{-6} \iint_{\mathbb{R}^3\times \mathbb{R}^3}\widehat{V}(p-q)\alpha(p)\alpha(q)\dd p\dd q 
\\& = \int_{\R^3} V(x) |\check \alpha_1(x)|^2 \dd x + \int_{\R^3} V(x) \check \alpha_1(x) \overline{\check \alpha_2(x)} \dd x
\\&+ \int_{\R^3} V(x) \check \alpha_2(x) \overline{\check \alpha_1(x)} \dd x + \int_{\R^3} V(x) |\check \alpha_2(x)|^2 \dd x
\\&= \int_{\R^3} V(x) |\check \alpha (x)|^2 \dd x,
\end{align*}
which is equality \eqref{alpha_conv_positive}.

Neglecting all terms aside \eqref{negative} and the one concerning the total density squared leads to the estimate
$$\mathcal F(\gamma, \alpha, \rho_0) \ge \frac12 \widehat V(0) \rho^2 + \rho_0 (2\pi)^{-3}\int_{\R^3} \widehat V(p) [\alpha(p) + \gamma(p)] \dd p.$$
By the inequality $\alpha^2\leq\gamma^2+\gamma$ we can deduce that
\begin{equation}
    \label{minus1/2}
    \alpha(p) + \gamma(p) \ge -\frac12.
\end{equation}
Using this fact we can easily find lower bound for the integral $\int \widehat{V}(\gamma+\alpha)$:
\begin{align*}
    \rho_0 \left(2\pi\right)^{-3}\int_{\mathbb{R}^3}\widehat{V}(p)\left[\gamma(p)+\alpha(p)\right]\dd p
    & \ge - \frac12\rho_0\left(2\pi\right)^{-3}\int_{\mathbb{R}^3}\widehat{V}(p)\dd p
     \\& = -\frac12V(0)\rho_0
     \\& \ge -\frac12V(0)\rho,
\end{align*}
where the last inequality follows from the assumption $V \ge 0$ (actually just the fact $V(0) \ge 0$, see Remark \ref{V_not_rep} at the end of this section). We deduce that for any state $(\gamma, \alpha, \rho_0) \in \mathcal D_{\rho}$ we have
\begin{equation}
\label{lower_bound}
\mathcal F(\gamma, \alpha, \rho_0) \ge \frac12 \widehat V(0) \rho^2 - \frac12 V(0) \rho.   
\end{equation}
Taking infimum over the domain $\mathcal{D_\rho}$ leads to obtaining the lower bound required for proving \eqref{asymptotics}.
\begin{remark}
Note that obtained estimation does not include any error term. This in particular means that the correction term in \eqref{asymptotics} is positive.
\end{remark}

\subsection{The upper bound}

We proceed to the derivation of the upper bound. Note that obtaining just the leading $\rho^2$ term is easy, it suffices to consider the trial state $(0, 0, \rho)$. The functional evaluated at it is equal to:
\begin{equation*}
    \mathcal{F}(0, 0, \rho_0 = \rho) = \frac12 \widehat{V}(0)\rho^2.
\end{equation*}
To obtain the second term $-\frac12 V(0)\rho$, we will construct a family of trial states that yields this result with an error of order $o(\rho)$.

\pagebreak

Before doing that, we shall give some intuition behind the construction. The main idea is to asympotically saturate the inequality \eqref{minus1/2} for $p$'s on sufficiently large set. As proven in \cite[Corollary 2.1]  {NapReuSol1} and \cite[Corollary 2]{Old}, the minimizing triple $(\alpha^{\min},\gamma^{\min},\rho^{\min}_{0})$ of the functional $\mathcal{F}$ corresponds to the quasi-free pure state, which means that it satisfies $\alpha^2(p) = \gamma(p)\left(\gamma(p)+1\right)$. This suggests considering the states $\gamma$ and $\alpha$ (up to a multiplicative constant) of the form:
\begin{equation*}
    \gamma = \lambda \chi_{A}, \quad \alpha = -\sqrt{\lambda^2 + \lambda}\chi_{A} ,
\end{equation*}
where $\chi_{A}$ is a characteristic function of a suitably chosen set $A$ of finite measure, in order to preserve the integrability of $\gamma$. Asymptotically for $\lambda \to \infty$ the expression $\gamma + \alpha$ behaves as:
\begin{equation}
\label{chosen_asymptotics}
\gamma(p) + \alpha(p) = \left(\lambda - \sqrt{\lambda^2 + \lambda}\right)\chi_{A}(p) = -\frac{1}{1 + \sqrt{1 + \frac1{\lambda}}}\chi_{A}(p) = \left(-\frac12 + O\left(\frac1{\lambda}\right)\right)\chi_{A}(p).
\end{equation}
In order to successfully use the above bound, the set $A$ needs to cover sufficiently large space. In this case the integral of $\widehat V$ on the complement of this set can be treated as a sufficiently small error term. Simultaneously, we need to remember about remaining terms of the functional. Their impact should be negligible with respect to total density $\rho$. 
To summarize, the requirements for the trial states $\gamma$ and  $\alpha$ are as follows:
\begin{itemize}
    \item both functions may have arbitrarily large value on their support,
    \item supports of those functions can cover arbitrarily large space,
    \item certain integrals in \eqref{functional} have to be small with respect to total density $\rho$.
\end{itemize}

Let us proceed to the construction. For $L > 0$ define set $P_L \subset \R^3$ as
$$P_L = \left[-\frac{L}{2},\frac{L}{2}\right]^3$$
and let $\chi_{P_L}$ be its characteristic function. Let $\gamma$ and  $\alpha$ be defined accordingly to the introduction of this section, that is as:
\begin{equation*}
    \gamma(p) = \left(2\pi\right)^{3}\lambda \chi_{P_L}(p), \quad \alpha(p) = -(2\pi)^3 \sqrt{\lambda^2+\lambda}\chi_{P_L}(p).
\end{equation*}
The relation between parameters $\lambda$, $L$ and total density $\rho$ is to be specified later, for this moment we will only assume that $\rho$ is sufficiently large so that $\rho_\gamma \le \rho$. As the density is fixed, we also have $\rho_0 = \rho - \rho_\gamma$.

Now, we will estimate each term in energy functional \eqref{functional} evaluated at the above trial state. First we will consider the kinetic term:
\begin{equation}
    \label{kinetic_est}
    \begin{split}
        \left(2\pi\right)^{-3}\int_{\mathbb{R}^3}p^2\gamma(p)\dd p = \lambda \int_{P_L}p^2\dd p = \lambda \int_{-L/2}^{L/2}\dd p_x\int_{-L/2}^{L/2}\dd p_y\int_{-L/2}^{L/2}\dd p_z\left(p_x^2+p_y^2+p_z^2\right)=\\
        = 3\lambda \int_{-L/2}^{L/2}\dd p_x\int_{-L/2}^{L/2}\dd p_y\int_{-L/2}^{L/2}p_z^2\dd p_z = 3\lambda L^2\left.\left(\frac{p_z^3}{3}\right)\right|_{-L/2}^{L/2} = \frac14 \lambda L^5.
    \end{split}
\end{equation}
Next, we will focus on the term $\rho_0 (2\pi)^{-3}\int\widehat{V}(\gamma+\alpha)$. Using \eqref{chosen_asymptotics} we obtain:
\begin{equation}
\label{decomposition}
\begin{split}
&\rho_{0}\left(2\pi\right)^{-3}\int_{\mathbb{R}^3}\widehat{V}(p)\left(\gamma_n(p)+\alpha_n(p)\right)\dd p = (\rho - \rho_\gamma) \int_{P_L}\widehat V(p) \left(-\frac12 + O(1/\lambda)\right) \dd p
\\& = - \frac12 \rho \int_{P_L} \widehat V(p) \dd p + \frac12 \rho_\gamma \int_{P_L} \widehat V(p) \dd p + (\rho - \rho_\gamma) O(1/\lambda) \int_{P_L} \widehat V(p) \dd p.
\end{split}
\end{equation}
The first term is the only negative one, we can decompose it further:
$$- \frac12 \rho \int_{P_L} \widehat V(p) \dd p = - \frac12 \rho \left(\int_{\R^3} \widehat V(p) \dd p - \int_{\R^3 \setminus P_L} \widehat V(p) \dd p\right) = - \frac12 \rho V(0) +  \frac12 \rho \int_{\R^3 \setminus P_L} \widehat V(p) \dd p.$$
The first expression in the above expansion is the term that we need in order to prove \eqref{asymptotics}. We will now show that the second term in the above expansion and the rest of terms in \eqref{decomposition} are negligible up to correct order. We have:
\begin{equation}
\label{est1}
\frac12 \rho \int_{\R^3 \setminus P_L} \widehat V(p) \dd p = \rho \cdot o(1)_{L \to \infty},
\end{equation}
where we used the facts $\widehat V \in L^1(\R^3)$, $\bigcup_L P_L = \R^3$ and Lebesgue Dominated Convergence Theorem. For the remaining terms in \eqref{decomposition} we get:
\begin{equation}
\label{est2}
\frac12 \rho_\gamma \int_{P_L} \widehat V(p) \dd p \le \frac12 \rho_\gamma \int_{\R^3} \widehat V(p) \dd p = \frac12 V(0) \rho_\gamma
\end{equation}
and
\begin{equation}
\label{est3}
(\rho - \rho_\gamma) O(1/\lambda) \int_{P_L} \widehat V(p) \dd p \le \rho \cdot O(1/\lambda)_{\lambda \to \infty} \cdot \int_{\R^3} \widehat V(p) \dd p = \rho \cdot O(1/\lambda)_{\lambda \to \infty}.
\end{equation}
Using \eqref{est1}, \eqref{est2} and \eqref{est3} to estimate \eqref{decomposition} we get:
\begin{equation}
\label{4_term_est}
\rho_{0}\left(2\pi\right)^{-3}\int_{\mathbb{R}^3}\widehat{V}(p)\left(\gamma_n(p)+\alpha_n(p)\right)\dd p \le -\frac12 V(0) \rho + \rho \cdot o(1)_{L \to \infty} + \frac12 V(0) \rho_\gamma + \rho \cdot O(1/\lambda)_{\lambda \to \infty}.
\end{equation}

It remains to bound the convolution term of $\mathcal{F}$. As we consider high density limit, we can assume $\lambda > 1$ and therefore $\gamma \ge 1$ on its support $P_L$. On this set we obtain pointwise estimate:
\begin{equation*}
    |\alpha| \le \sqrt{\gamma(\gamma + 1)} = \gamma \sqrt{1 + 1/\gamma} \le \sqrt2 \gamma.
\end{equation*}
For $p$'s outside $P_L$ the above inequality is also satisfied in a trivial way. Now we can make the estimation:
\begin{equation}
\label{convo_est}
        \begin{gathered}
            \frac12\left(2\pi\right)^{-6}\iint_{\mathbb{R}^3\times\mathbb{R}^3}\widehat{V}(p-q)\left[\gamma(p)\gamma(q)+\alpha(p)\alpha(q)\right]\dd p\dd q \leq \\
            \leq \frac32\left(2\pi\right)^{-6}\iint_{\mathbb{R}^3\times\mathbb{R}^3}\widehat{V}(p-q)\gamma(p)\gamma(q)\dd p\dd q  \\
            = \frac32\left(2\pi\right)^{-6}\int_{\mathbb{R}^3}\left(\widehat{V}\ast\gamma\right)(p)\gamma(p)\dd p\leq \frac32\left(2\pi\right)^{-6}\|\widehat{V}\ast\gamma\|_\infty\int_{\mathbb{R}^3}\gamma(p)\dd p \\
            = \frac32\left(2\pi\right)^{-3}\|\widehat{V}\ast\gamma\|_\infty\rho_{\gamma} \leq \frac32\left(2\pi\right)^{-3}\|\widehat{V}\|_\infty \|\gamma\|_1 \rho_{\gamma} = \frac32\|\widehat{V}\|_\infty\rho_{\gamma}^2,
        \end{gathered}
\end{equation}
where in the last line we used Young inequality for convolution.

We now combine all of the established estimates \eqref{kinetic_est}, \eqref{4_term_est} and \eqref{convo_est}. Estimating $\rho_\gamma \le \rho^2_\gamma = \lambda^2 L^6$ and denoting constants independent of the parameters in the error term by $C_1$ and $C_2$ we obtain a bound:
\begin{equation}
\label{upper_bound_parameters}
\begin{split}
\mathcal{F}(\gamma, \alpha, \rho_0) &\le \frac12 \widehat V(0) \rho^2 - \frac12 V(0) \rho 
\\&+ C_1 \lambda L^5 + \rho \cdot \big(o(1)_{L \to \infty} + O(1/\lambda)_{\lambda \to \infty}\big) +  C_2 \lambda^2 L^6.
\end{split}
\end{equation}
We will specify $\lambda$ and $L$ to get desired result. We will look for relation of the form:
$$\lambda = \rho^{r}, \; L = \rho^s,$$
for appropriately chosen powers $r > 0$ and $s > 0$. With this selection, the error term (as $\rho \to \infty$) is of order: 
\begin{equation}
\label{rs_assuptions}
C_1 \rho^{r + 5s} + \rho \big(o(1) + O(\rho^{-r})\big) + C_2 \rho^{2r + 6s} = o(\rho) + O(\rho^{2r + 6s}).
\end{equation}
Now, choose any parameters $r$ and $s$ such that $2r + 6s < 1$. Then the error term is of order $o(\rho)$ and therefore:
$$\mathcal{F}(\gamma, \alpha, \rho_0 = \rho - \rho_\gamma) \le \frac12 \widehat V(0) \rho^2 - \frac12 V(0) \rho + o(\rho)_{\rho \to \infty}.$$
Combining it with lower bound \eqref{lower_bound} we have finished the proof of Theorem \ref{main_theorem}.

\pagebreak

We will end this section with some discussion concerning the obtained result.

\begin{remark}
We shall briefly discuss the structure of the trial states realizing the optimal upper bound of the ground state energy. First note that the kinetic energy does not explicitly contribute to the energy expansion \eqref{asymptotics} as both terms are related to the interaction potential $V$. This observation is with agreement with trial states used in the proof, where the kinetic energy contribution was a part of the error term and it becomes negligible in the high density limit.

Another observation is that for the trial states the ratio $\rho_0/\rho$ of the condensate density and total density of the system tends to one in this limit. It follows that, in this limit, it is energetically favorable for particles to almost fully occupy the zero momentum state.
\end{remark}

\begin{remark}
\label{V_not_rep}
The essential estimates in the above proof were obtained using only the assumption $\widehat V \ge 0$ (and its consequence $V(0) \ge 0$) without considering the sign of $V$. The assumption $V \ge 0$ was the technical one, necessary in \eqref{alpha_conv_positive} to bound the quadratic term involving $\alpha$. As the quadratic term involving both $\gamma$ and $\alpha$ in the construction of the functional arises from evaluating the expectation value of $\widehat V$ in a quasi-free state, it is expected that only the assumption $\widehat V \ge 0$ should be sufficient. Similar observation on the potential was also made in \cite{Lieb}. This suggest that the estimates concerning the lower bound of $\mathcal{F}$ (in particular \eqref{alpha_conv_positive}) might be refined by obtaining an estimate of the entire convolution term of the functional.
\end{remark}

\begin{remark}
\label{other_dim}
Presented method is also applicable for dimensions other than $d=3$. Repeating the argument for any dimension $d$, the analogue of \eqref{rs_assuptions} is an error of order:
$$o(\rho) + O\left(\rho^{m(r,s)}\right),$$
where
$$m(r,s) = \max \{r + (d+2)s, 2r + 2ds\}.$$
Choosing $r,s$ such that $m(r,s)<1$ yields the result. Note that the maximum of the above two expressions in nontrivial only if $d = 1.$
\end{remark}

\section{Proof of Theorem \ref{extra_regularity} and further discussion}

In this section we shall prove that adding some extra assumptions on the potential $V$ (or rather its Fourier transform $\widehat V$) leads to obtaining lower order error term in \eqref{asymptotics}. Note that the worst error estimate in the proof of Theorem \ref{main_theorem} comes from \eqref{est1} when estimating the integral of $\widehat V$ on the complement of the set $P_L$. This justifies the need of assumption on the rate of decay of $\widehat V$ at infinity.

\begin{proof}[Proof of Theorem \ref{extra_regularity}]
The proof is divided into two steps, each one concerning particular assumption on $\widehat V$.

\textbf{Step 1}
\\*We repeat every step in the proof of the previous theorem. With the extra assumption on $\widehat V$ we can make further estimation in \eqref{est1} and write:
\begin{equation}
\label{decay}
\frac12 \rho \int_{\R^3 \setminus P_L} \widehat V(p) \dd p \le \frac12 \rho \int_{\R^3 \setminus P_L} \frac{C}{|p|^k} \dd p \le  \tilde C \rho \int_{L/2}^\infty \frac{r^2}{r^k} \dd r = \rho \cdot O\left(\frac{1}{L^{k-3}}\right)_{L \to \infty}.
\end{equation}
The analogue of the upper bound \eqref{upper_bound_parameters} becomes as follows:
\begin{equation*}
\begin{split}
\mathcal{F}(\gamma, \alpha, \rho_0) &\le \frac12 \widehat V(0) \rho^2 - \frac12 V(0) \rho 
\\&+ C_1 \lambda L^5 + \rho \cdot \left(O\left(1/L^{k-3}\right)_{L \to \infty} + O(1/\lambda)_{\lambda \to \infty}\right) +  C_2 \lambda^2 L^6.  
\end{split}
\end{equation*}
Once again we will look for relation of the parameters of the form:
$$\lambda = \rho^{r}, \; L = \rho^s$$
for appropriately chosen powers $r > 0$ and $s > 0$. The error term is then of order:
$$C_1 \rho^{r + 5s} + \rho \Big(O\left(\rho^{-(k-3)s}\right) + O(\rho^{-r})\Big) + C_2\rho^{2r + 6s} = O(\rho^{m(r,s)}),$$
where
$$m(r,s) = \max \{2r + 6s,  1 - (k-3)s, 1 - r\}.$$
To get the error of desired order, parameters $r$ and $s$ need to be chosen in a way that $m(r,s) \le \theta$.

First, we will derive the necessary conditions for the existence of such $\theta$. Assume there exist $r > 0$ and $s > 0$ satisfying $m(r,s) \le \theta$. Then, combining inequalities $1-r \le \theta$ and $2r + 6s \le \theta$ results in
\begin{equation}
\label{theta_1st_est}
\theta > \frac23,    
\end{equation}
so this is the lowest order of error possible to obtain by proposed method (note that this inequality is strict). Next, as $r \ge 1 - \theta$, we can write:
$$r = 1 - \theta + \delta,$$
for some $\delta \ge 0$ and plug it to the first inequality, that is:
$$2r + 6s = 2 - 2\theta + 2 \delta + 6s \le \theta \Longleftrightarrow s\le \frac{3\theta - 2 - 2\delta}{6}.$$
Using the fact that $(k-3)s \ge 1 -\theta$, the following inequality must be fulfilled:
$$\frac{1 - \theta}{k-3} \le \frac{3\theta - 2 - 2 \delta}{6} \Longleftrightarrow \delta \le \frac{3\theta (k - 1) - 2k}{2(k-3)}.$$
As $\delta \ge 0$ we need to have:
$$3\theta (k - 1) - 2k \ge 0 \Longleftrightarrow \theta \ge \frac{2k}{3k - 3}.$$
Note that the first obtained estimate \eqref{theta_1st_est} is the limiting case (for $k \to \infty$) of the above inequality.

Reversing the order of this reasoning we observe that derived condition for $\theta$ is also the sufficient one for the existence of $r$ and $s$ yielding $m(r,s) \le \theta$. One of possible choices of such $r$ and $s$ is
\begin{equation*}
r = 1 - \theta, \; \; s = \frac{3\theta - 2}{6}.   
\end{equation*}
Thus the expansion \eqref{asymptotics_regular} stated in the Theorem \ref{extra_regularity} is valid for $\theta = \frac{2k}{3k - 3}$.

\textbf{Step 2}
\\*When $\widehat V$ decays exponentially, we can do analogous computation as in \eqref{decay} and obtain
$$\frac12 \rho \int_{\R^3 \setminus P_L} \widehat V(p) \dd p \le \frac12 \rho \int_{\R^3 \setminus P_L} Ce^{-c|p|} \dd p \le  C' \rho \int_{L/2}^\infty r^2 e^{-cr} \dd r \le C'' \rho L^{2}e^{-cL} = \rho \cdot O\left(L^2e^{-cL}\right)_{L \to \infty}$$
and therefore
\begin{equation*}
\begin{split}
\mathcal{F}(\gamma, \alpha, \rho_0) &\le \frac12 \widehat V(0) \rho^2 - \frac12 V(0) \rho 
\\&+ C_1 \lambda L^5 + \rho \cdot \bigg(O\left(L^2e^{-cL}\right)_{L \to \infty} + O(1/\lambda)_{\lambda \to \infty}\bigg) +  C_2 \lambda^2 L^6.  
\end{split}
\end{equation*}
This time we will choose parameters $\lambda$ and $L$ as
$$\lambda = \rho^{1/3}, \;  L = \frac{1}{3c}\ln \rho,$$
where $c$ is the same constant as in the exponent in asymptotic term. With this selection we can express the error term as:
$$\tilde C_1\rho^{1/3}(\ln \rho)^5 + \rho \left(O\left(\frac{(\ln \rho)^2}{\rho^{1/3}}\right) + O\left(\frac1{\rho^{1/3}}\right)\right) + \tilde C_2\rho^{2/3}(\ln \rho)^6 = O\left(\rho^{2/3}(\ln \rho)^6\right).$$
\end{proof}

We shall end with some general remarks.

\begin{remark}
Similarly like in the Remark \ref{other_dim} it is possible to generalize this result to dimensions other than $d=3$. In the polynomial decay case the assumption  \eqref{rate_of_conv} needs to hold for $k > d$. Following the same reasoning as above, we deduce that $\theta$, for which the result is valid, can be chosen as:
$$\theta = \frac{2k}{3k - d},$$
which is obtained for parameters $r$ and $s$ chosen as, for example,
$$r = 1 - \theta, \; \; s = \frac{3\theta - 2}{2d}.$$
The limiting value ($k \to \infty$) is the same as before, equal to $\frac23$.

In the case of exponential decay, the asymptotic error will be of order $O\left(\rho^{2/3}(\ln \rho)^{2d}\right)$, as the logarithm term is related to the squared volume of the set $P_L$, which is equal to $L^{2d}$.
\end{remark}

\begin{remark}
We would like to note that considered assumptions on the polynomial or exponential rate of decay are very general, as they hold for many reasonable potentials (e.g. for $V$ being a Schwartz function). Presented method, however, can be easily generalized to rates of decay of $\widehat V$ of other form. 

Also note that, by using similar argument preceding inequality \eqref{theta_1st_est}, regardless of the rate of decay it is not possible by this method to obtain error of order lower than $\rho^{2/3}$. As the considered trial state appears to be well suited for this problem, this suggest that the explicit form of the third leading order term in the expansion of the ground state energy might be much more dependent on the properties of the interaction potential $V$.
\end{remark}

\noindent {\bf Acknowledgments.} We want to thank Prof. Marcin Napi\'orkowski for many helpful discussions and support. The work of NM was supported by the Polish-German NCN-DFG grant Beethoven Classic 3 (project no. 2018/31/G/ST1/01166).

\section*{Data availability}
Data sharing is not applicable to this article as no new data were created or
analyzed in this study.

\end{document}